\documentclass{amsart}

\usepackage[usenames]{color}
\usepackage{enumerate}
\usepackage{amsmath}
\usepackage{amsthm}
\usepackage{amsfonts}
\usepackage{float}
\usepackage[colorlinks=true,
linkcolor=webgreen,
filecolor=webbrown,
citecolor=webgreen]{hyperref}
\usepackage{url}
\usepackage{bclogo}
\usepackage{tikz}
\usepackage{algorithm}
\usepackage[noend]{algpseudocode}

\theoremstyle{plain}
\newtheorem{theorem}{Theorem}
\newtheorem*{theorem*}{Theorem}
\newtheorem{corollary}[theorem]{Corollary}
\newtheorem*{corollary*}{Corollary}
\newtheorem{lemma}[theorem]{Lemma}
\newtheorem*{lemma*}{Lemma}
\newtheorem{proposition}[theorem]{Proposition}
\newtheorem*{proposition*}{Proposition}

\theoremstyle{definition}

\newtheorem*{definition*}{Definition}

\newtheorem*{example*}{Example}

\newtheorem*{conjecture*}{Conjecture}

\newtheorem*{problem*}{Problem}

\newtheorem*{todo*}{TO DO}

\newtheorem*{question*}{Question}

\theoremstyle{remark}

\newtheorem*{remark*}{Remark}

\newtheorem*{notation*}{Notation}

\definecolor{webgreen}{rgb}{0,.5,0}
\definecolor{webbrown}{rgb}{.6,0,0}

\def\modd#1 #2{#1\ ({\rm mod}\ #2)}

\begin{document}

\title{Counting Small Permutation Patterns}

\author{Chaim Even-Zohar}
\address{Chaim Even-Zohar, The Alan Turing Institute, London, NW1 2DB, UK}
\email{chaim@ucdavis.edu}

\author{Calvin Leng}
\address{Calvin Leng, University of Southern California, Los Angeles, CA 90089, US}
\email{cleng@usc.edu}

\maketitle

\begin{abstract}
A sample of $n$ generic points in the $xy$-plane defines a permutation that relates their ranks along the two axes. Every subset of $k$ points similarly defines a pattern, which occurs in that permutation. The number of occurrences of small patterns in a large permutation arises in many areas, including nonparametric statistics. It is therefore desirable to count them more efficiently than the straightforward $\tilde{O}(n^k)$ time algorithm. 

This work proposes new algorithms for counting patterns. We show that all patterns of order 2 and 3, as well as eight patterns of order 4, can be counted in nearly linear time.  To that end, we develop an algebraic framework that we call \emph{corner tree formulas}. Our approach generalizes the existing methods and allows a systematic study of their scope.

Using the machinery of corner trees, we find twenty-three independent linear combinations of order-4 patterns, that can be computed in time $\tilde{O}(n)$. We also describe an algorithm that counts one of the remaining 4-patterns, and hence all 4-patterns, in time~$\tilde{O}(n^{3/2})$.

As a practical application, we provide a nearly linear time computation of a statistic by Yanagimoto (1970), Bergsma and Dassios (2010). This statistic yields a natural and strongly consistent variant of Hoeffding's test for independence of $X$ and $Y$, given a random sample as above. This improves upon the so far most efficient $\tilde{O}(n^2)$ algorithm.
\end{abstract}

\section{Introduction}\label{intro}

In many situations it is interesting to investigate the \emph{local} structure of a large combinatorial object. This may mean small subgraphs of a given large graph, short patterns in a long sequence, and so forth. Local views provide meaningful and efficient ways to analyze large combinatorial structures, both in theory and in practice. In the realm of big data, where data sets are too large to observe in full, local approaches are essentially inescapable.

Here we investigate the local structure of permutations. The restriction of a large permutation to some of its entries yields a \emph{pattern}, a smaller permutation defined by the relative ordering of the observed entries. Namely, the order-$k$ pattern $\sigma$ occurs at positions $i_1 <i_2<\dots<i_k$ of the order-$n$ permutation $\pi$, if $\sigma(j)<\sigma(j')$ whenever $\pi(i_j) < \pi(i_{j'})$ for $j \neq j'$. For example, the pattern \emph{132} occurs in $2\underline{3}64\underline{7}\underline{5}1$ at the marked positions~$\{2,5,6\}$. We throughout use the one line notation $\pi = \pi(1)\pi(2)\cdots\pi(n)$ for permutations and their patterns.

We refer to the positions $i_1 <i_2<\dots<i_k$ as an \emph{occurrence} of $\sigma\in S_k$ in $\pi\in S_n$ and denote the number of occurrences of $\sigma$ in $\pi$ by $\#\sigma(\pi)$. For example, one can easily verify that $\#\textit{132}\,(2364751)=7$. The \emph{$k$-profile} of $\pi \in S_n$ consists of the $k!$~numbers $\#{\sigma}(\pi)$ for all $\sigma \in S_k$. These numbers clearly sum up to~$\tbinom{n}{k}$.

\medskip
There is considerable combinatorial literature concerning the numbers $\#{\sigma}(\pi)$. For example, given $\sigma$, one seeks the maximum of $\#{\sigma}(\pi)$ over all $\pi \in S_n$~\cite{price1997packing,albert2002packing,hasto2002packing, presutti2010packing,burstein2010packing,balogh2015minimum,sliacan2017improving}. It is also interesting to understand the distribution of $\#{\sigma}(\pi)$ for a randomly chosen $\pi \in S_n$~\cite{fulman2004stein,bona2007copies,bona2010three,burstein2010packing,janson2015asymptotic,hofer2017central,even2018patterns}. Study of $k$-profiles leads to the notion of ``permutons'', limits of sequences of permutations~\cite{hoppen2013limits,glebov2015finitely,kenyon2019permutations}. This theory includes the case of quasi-random permutations~\cite{cooper2004quasirandom,cooper2006permutation,kral2013quasirandom}.

But the local viewpoint on permutations is of interest as well in a range of applications, most notably in nonparametric statistics of bivariate data. Consider a sequence of $n$ independent paired samples $(X_i,Y_i)$ drawn from a common continuous distribution $(X,Y)$ on~$\mathbb{R}^2$. Various measures have been suggested to detect and quantify a relation between $X$ and $Y$ based on these samples. Often one prefers to rely only on the ranking of $X_i$ and $Y_i$ rather than on their numerical values. In this case it suffices to consider the permutation $\pi \in S_n$, defined by $\mathrm{rank}\,Y_i=\pi(\mathrm{rank}\,X_i)$, where the rank of $X_1,\dots,X_n$ is their order preserving map to $\{1,\dots,n\}$.

Many nonparametric measures depend only on $\pi$'s local structure. Here are some classical examples. Kendall's $\tau$ correlation coefficient is based on the 2-profile: $\tau = \left[\#\textit{12}\,(\pi) - \#\textit{21}\,(\pi)\right]/\tbinom{n}{2}$~\cite{kendall1938new}. Similarly Spearman's $\rho$ is a function of the 3-profile of~$\pi$~\cite{spearman1904proof}, and so is a rotation-invariant measure of correlation by Fisher and Lee~\cite{fisher1982nonparametric,janson1984asymptotic}, and other correlation tests~\cite{crouse1966distribution}.

Local features of $\pi$ can even detect whether there is \emph{any} relationship between $X$ and~$Y$, as demonstrated by Hoeffding's independence test~\cite{hoeffding1948non}. This classical nonparametric test is based on the 5-profile of $\pi$, and so is a refined version with wider consistency~\cite{blum1961distribution,rosenblatt1975quadratic,de1980cramer}. A simpler variant that only relies on the 4-profile has been proposed by Bergsma and Dassios~\cite{bergsma2010nonparametric,bergsma2010consistent,bergsma2014consistent}, and goes back to a result of Yanagimoto~\cite{yanagimoto1970measures, drton2018high}. A~general family of independence tests expressible by pattern counts is suggested in~\cite{heller2016consistent}.

Pattern counting also plays a role in \emph{property testing} and \emph{parameter testing} for permutations~\cite{hoppen2011testing}. In these problems, one estimates features of a large permutation using a randomized algorithm that makes a small number of queries.

Applications to stack-sorting~\cite{knuth1968fundamental} have inspired a significant amount of research around the case $\#{\sigma}(\pi)=0$. In particular, there has been interest in the count of such $\sigma$-avoiding permutations of order~$n$~\cite{simion1985restricted,marcus2004excluded}. This line of research in Enumerative Combinatorics has introduced a wealth of generalizations to other permutation statistics, and  connections to other branches of Mathematics~\cite{steingrimsson2010generalized,kitaev2011patterns,bona2012combinatorics}. 

The extensive study of $\sigma$-avoidance has naturally led to the exploration of its algorithmic aspects. The most obvious question is \emph{permutation pattern matching}: deciding whether a given pattern $\sigma \in S_k$ occurs anywhere in a given permutation $\pi \in S_n$. In this generality the problem is NP-complete~\cite{bose1998pattern,bliznets2015kernelization,cooper2016complexity}. Straightforward exhaustive search solves the problem in time $\tilde{O}(\tbinom{n}{k}\tbinom{k}{2})$, but faster algorithms have been found, improving the exponent by a constant factor~\cite{ahal2008complexity,bruner2016fast,kozma2019faster,berendsohn2019finding}. Numerous works have achieved improved running times for patterns $\sigma$ that are either short or have some structural properties~\cite{schensted1961longest,chang1992efficient,ibarra1997finding,bose1998pattern,albert2001algorithms,yugandhar2005parallel,guillemot2009pattern,bruner2012fast,bruner2013computational,bruner2016fast,guillemot2014finding,han2014parallel,jelinek2017hardness,han2018algorithms,berendsohn2019complexity}. Remarkably, this problem is ``fixed-parameter tractable''. For $\sigma \in S_k$ and large $n$ it can be solved in time $\tilde{O}(f(k)\cdot n)$ \cite{guillemot2014finding}.     

The question of \emph{permutation pattern counting} seeks to determine $\#{\sigma}(\pi)$, the number of occurrences of a given pattern $\sigma \in S_k$ in a given permutation $\pi \in S_n$. This problem, which is $\#P$-complete in general, has received less attention. It has mostly been considered in papers on the decision problem, as it sometimes happens that a decision algorithm can be adapted to address as well the counting problem~\cite{bose1998pattern,albert2001algorithms,yugandhar2005parallel,berendsohn2019finding}. For example, a pattern $\sigma$ from the special class of \emph{separable} permutations can be counted in time $\tilde{O}(kn^6)$~\cite{bose1998pattern}. However, the general linear time algorithm from~\cite{guillemot2014finding} does not yield a counting method. A recent paper~\cite{berendsohn2019finding} shows that if $\#{\sigma}(\pi)$ can be computed for all $\sigma$ and $\pi$ in time $f(k)n^{o(k/\log k)}$, then the exponential-time hypothesis fails to hold. For comparison, recall that the run time of the trivial algorithm is~$n^{O(k)}$.

While the above hardness results virtually settle the counting problem for general $\sigma \in S_k$ and asymptotically in~$k$, they leave open some questions on specific patterns and particular~$k$. For instance, which patterns can be counted in linear time, and how? Since the decision algorithm of ~\cite{guillemot2014finding} does not help us here, new ideas are clearly needed.

From a practical standpoint, these questions are much relevant to the implementation of rank-based statistical tests. Spearman's $\rho$ is computable in linear time by definition. Kendall's~$\tau$ naively requires $\tilde{O}(n^2)$ time, but an algorithm by Knight computes it in $\tilde{O}(n)$~\cite{knight1966computer}. By the same technique, Hoeffding's independence test is readily computable in $\tilde{O}(n)$, despite being stated and sometimes implemented with $\tilde{O}(n^2)$ formulas~\cite{hoeffding1948non,harrell2006hmisc,hollander2013nonparametric}. Several works have addressed the computation of the Bergsma--Dassios--Yanagimoto test, only bringing it down to $\tilde{O}(n^2)$ time~\cite{bergsma2010consistent,bergsma2014consistent,weihs2016efficient,heller2016computing,weihs2018symmetric}. It seems that these pattern-based statistics have only been treated ad hoc, by dedicated algorithms, without a unified or systematic approach. 

\medskip
This work proposes new approaches for permutation pattern counting. In Section~\ref{trees} we define the key notion of a \emph{corner tree}, a structure that lets us apply dynamic programming to the problem. This leads to a class of permutation statistics that have a \emph{corner tree formula}. We present an algorithm that evaluates such statistics efficiently, proving the following theorem.

\begin{theorem}\label{main}
Let $f:S_n \to \mathbb{Z}$ be a permutation statistic, given as a corner tree formula. There exists an algorithm that computes $f(\pi)$ for a given permutation $\pi \in S_n$ in running time $\tilde{O}(n)$.
\end{theorem}

Although corner trees can come in any size, they prove particularly useful for counting small patterns. In Section~\ref{small}, we give corner tree formulas for all 3-patterns, which implies the following. 

\begin{corollary}\label{three}
The number of occurrences of any pattern of size $3$ in a given permutation $\pi \in S_n$ can be computed in $\tilde{O}(n)$ time.
\end{corollary}

Similarly, we show that corner tree formulas exist for the following eight patterns of size~4.

\begin{corollary}\label{the-eight}
The occurrences of any of the patterns 1234, 1243, 2134, 2143, 3412, 3421, 4312, 4321 in a given permutation $\pi \in S_n$ can be counted in time $\tilde{O}(n)$.
\end{corollary}

The total count of these eight patterns is sufficient for computing the Bergsma--Dassios--Yanagimoto statistic. This yields the following practical consequence.

\begin{corollary}\label{bd}
The Bergsma--Dassios--Yanagimoto independence test for continuous bivariate data can be computed in $\tilde{O}(n)$ time.
\end{corollary}

In section~\ref{small} we provide a review of this increasingly popular statistical test. Corollary~\ref{bd} is a significant improvement over the best previous method, which requires $\tilde{O}(n^2)$ time and space~\cite{heller2016computing}. Indeed, a near linear time performance is crucial when dealing with large amounts of data. We thus expect our method to be beneficial in real world scenarios. 

We derive explicit procedures for fast evaluation of this statistic. See Proposition~\ref{bdsym}. The paper is accompanied by a Python package that includes, among other things, a routine for the Bergsma--Dassios--Yanagimoto statistic's computation~\cite{even2019corners}.

\medskip
Another important ingredient of our approach is treating the $k$-profile vector as a whole, rather than looking on each pattern separately. The permutation statistics that we study are linear combinations of pattern counts:
$$ F(\pi) \;=\; \sum\limits_{\sigma}^{~}\,f_{\sigma}\,\#{\sigma}(\pi) $$ 
Here the sum is over patterns, with finitely many nonzero $f_{\sigma} \in \mathbb{Q}$. The corner tree formulas that we define span a vector space of such combinations. This viewpoint provides much more information on the $k$-profile, as demonstrated in the following proposition with $k=4$.

\begin{proposition}\label{23}
There exist twenty-three linearly independent combinations,
$$ F_i(\pi) \;=\; \sum\limits_{\sigma \in S_4}^{~}\,f_{i\sigma}\,\#{\sigma}(\pi) \;\;\;\;\;\;\;\;\;\;\;\; i \in \{1,2,\dots,23\} $$ 
that can be evaluated for a given $\pi \in S_n$ in time $\tilde{O}(n)$.
\end{proposition}

Since the 4-profile is a 24-dimensional vector, only one additional combination is needed to fully reveal it. In other words, one can efficiently compute this vector up to a single linear degree of freedom. All the remaining problems of counting 4-patterns are thus equivalent to each other.

Clearly, Proposition~\ref{23} was established with computer assistance. At the end of Section~\ref{small} we provide more details about the implementation.

Our code provides general-purpose routines for finding and using corner tree formulas, and it is scalable to larger patterns and combinations. For example, it yields a 100-dimensional space of linear functions of the 120-dimensional 5-profile, that are computable in $\tilde{O}(n)$ time.

\medskip
In Section~\ref{four} we explore new counting methods beyond corner trees. We naturally focus on the missing piece of the 4-profile. An $\tilde{O}(n^2)$ time computation is straightforward from the existing methods. We present two algorithms that do somewhat better, as follows.

\begin{theorem}\label{n53}
The occurrences of any 4-pattern in a given permutation $\pi \in S_n$ can be counted in $\tilde{O}(n^{5/3})$ time and $\tilde{O}(n)$ space.
\end{theorem}

\begin{theorem}\label{n32}
The occurrences of any 4-pattern in a given permutation $\pi \in S_n$ can be counted in $\tilde{O}(n^{3/2})$ time and space.
\end{theorem}

Finally, in Section~\ref{discuss}, we discuss the scope of our methods, and questions for future work.

\begin{remark*}
Since our results reduce the degree of some polynomial run times, we are less concerned with logarithmic factors. Therefore, we throughout use the soft~O notation, $\tilde{O}(f(n))$, which means $O(f(n)\,\log^cn)$ where $c$ is some constant. In words, we sometimes write ``nearly linear'', ``nearly quadratic'', and so on.

Moreover, some algorithms in the literature, such as~\cite{guillemot2014finding}, already assume a computational model that operates on $(\log n)$-bit numbers in $O(1)$ time. Thus, by ignoring logarithmic factors, our statements become more robust to different choices of the abstract machine.

In contrast, some results in this area are machine dependent. In the case of counting the pattern \emph{12}, an improvement of order $\sqrt{\log n}$ was given specifically assuming a Word RAM model~\cite{chan2010counting}. 

Since the logarithmic factor in $\tilde O(n)$ may be of interest in practical applications, we note that the algorithm of Theorem~\ref{main} and its corollaries actually performs $O(n \log n)$ operations on $O(\log n)$-digit numbers. For more details see the very end of Section~\ref{trees}.
\end{remark*}

\section{Trees}\label{trees}

Fix a pattern $\sigma \in S_k$. Given a permutation $\pi \in S_n$ as input, our goal is to compute the pattern count $\#{\sigma}(\pi)$, also denoted $\#\sigma$ whenever no confusion arises. Simple $\tilde{O}(n)$ time algorithms are known for \#\emph{12}, and extend to larger increasing patterns \#\emph{12...k}. We first review these techniques, as they are relevant to our general approach.

The task of computing Kendall's $\tau$ rank correlation coefficient is essentially, after sorting the samples, counting either increasing or decreasing pairs in a permutation.
$$ \tau \;=\; \frac{\#\emph{12}-\#\emph{21}}{\tbinom{n}{2}} \;=\; 2 \cdot \frac{\#\emph{12}\,}{\tbinom{n}{2}} - 1 $$
A naive implementation iterates through every two entries and compares their values, which yields an answer in $O(n^2)$ steps.

Knight improved this to $O(n \log n)$ steps with a divide-and-conquer approach, motivated by merge sort~\cite{knight1966computer}. While performing the usual merge sort on a permutation, treating it as a list to be ordered, we record some numbers that add up to~\#\emph{12}. Namely, when we merge two sub-lists and select an element from the second one, we accumulate the number of elements selected so far from the first one. Every occurrence of \emph{12} contributes one to the result.

A different approach uses dynamic programming~\cite[ex 14.1-7]{cormen2009introduction}. We iterate once through $\pi$. At step $x$, we \emph{insert} the number 1 at position $y=\pi(x)$ of an array~$A$, and we add to the result its current \emph{prefix sum} of $y-1$ entries. For example, in step 4 of $\#\!\textit{12}(28176453)$, the array $A$ is updated $11000001 \to 11000011$ since $\pi(4)=7$, and the result increases by $2$, which is the sum of~$A[1]$...$A[6]$. 

We briefly describe the  \emph{sum tree} data structure~\cite{shiloach1982n2log}, that performs these two operations in logarithmic time, as it is simple and will serve us well later. Place the array entries at the leaves of a complete binary tree of depth $\lceil \log n \rceil$, from left to right. Insertion at $y$ is done by an increment to the $y$-th leaf and to all its ancestors. Every node hence always contains the sum of its children. The \emph{$y$-prefix}, i.e.~the first $y-1$ leaves, may be efficiently summed using left-siblings of $y$'s ancestors. 

In fact, the two presented methods add exactly the same numbers in a different order. We also remark that similar tree-based structures are used in~\cite{albert2001algorithms} for permutation pattern matching. 

Counting larger increasing patterns, such as \#\emph{123}, is similarly done with more sum trees~\cite{chaturvedi2018geeks}. Let the sum tree $A$ record the already-observed values as before. In the next sum tree $B$, we insert at $B[y]$ the number of \emph{12} such that $y$ corresponds to the~\emph{2}. For example, $B[5]=3$ for~$\pi=\underline2\,8\,\underline1\,7\,6\,\underline{4}\,\underline{\underline{5}}\,3$, where the three \emph{12}-occurrences with values $(*,5)$ are marked. We remark that the array~$B$ eventually sums to~\#\emph{12}, and that $\pi$ can be reconstructed from $B$, which is closely related to the \emph{inversion table} of~$\pi$~\cite[page 36]{stanley2011enumerative}, but these facts are not needed here.

Similar to the computation of \#\emph{12} by summing queries to~$A$, the computation of \#\emph{123} sums queries to~$B$. To sum up, at every step~$x$ the algorithm will

\begin{samepage}
\begin{itemize}
\item[(a)] insert one to $A[\pi(x)]$,
\item[(b)] insert the $\pi(x)$-prefix sum of $A$ to $B[\pi(x)]$,
\item[(c)] add the $\pi(x)$-prefix sum of $B$ to the result. 
\end{itemize}
\end{samepage}

An occurrence of \emph{123} at $i_1<i_2<i_3$ contributes one to $A[\pi(i_1)]$ at (a) of step~$i_1$, which we pass along to $B[\pi(i_2)]$ at (b) of step $i_2$, and to the result at (c) of step $i_3$. 

One can similarly compute  \#\emph{12...k} in $\tilde{O}(k^2n)$ time and space with $k-1$ sum trees. Every sum tree $A_m$ will hold counts of \emph{12...m} based on prefix sums of~$A_{m-1}$. Note that this algorithm can also yield \#\emph{k...21}, by reversing the input $\pi$.

\bigskip \noindent \textbf{Corner Trees} \medskip

The graph of an order-$n$ permutation $\pi$ in the $xy$-plane consists of $n$ points. In the above dynamic programming, one scans these points from left to right, and queries a $\pi(x)$-prefix sum in order to integrate some function over the points to the left and below $(x,\pi(x))$. From this perspective, the extension to other statistics will be done by looking at all four quadrants around a point, combining several of them together, and repeating recursively. These ideas motivate the following definitions.

\def\root{\raisebox{-4pt}{\bcrosevents}}

\begin{definition*}
A \emph{corner tree} is a rooted tree, whose vertices are labeled by 
$$\{\root,\; \texttt{SW},\; \texttt{SE},\; \texttt{NW},\; \texttt{NE}\}$$
where $\root$ is designated as the label for the root.
\end{definition*}

\begin{example*}
See Figure~\ref{atree}.
\end{example*}

\newcommand{\ct}[1]{\raisebox{-5pt}{\tikz{
\tikzstyle{vertex}=[circle, minimum size=16pt, inner sep=-10pt, draw];
\foreach \label/\x/\y/\parent in {#1} {
\node[vertex] (\x\y) at (\x,\y) {\texttt{\label}};
\draw[-] (\x\y)--(\parent);
}}}}

\begin{figure}[h]
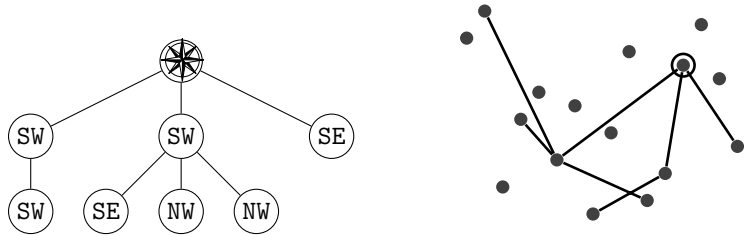

\centering
\ct{\root/5/5/55, SW/3/4/55, SW/5/4/55, SE/7/4/55, SE/4/3/54, NW/5/3/54, NW/6/3/54, SW/3/3/34} \;\;\;\;\;\;\;\;\;\;\;\;
\tikz[scale=0.2]{
\tikzstyle{vertex}=[circle, minimum size=5pt, inner sep=0pt, fill, darkgray];
\foreach \x/\y in {1/14,2/16,3/3,4/8,5/10,6/5,7/9,8/1,9/7,10/13,11/2,12/4,13/12,14/15,15/11,16/6} {
\node[vertex] (\x) at (\x*1.2,\y*0.9) {};}
\draw[-, line width=1pt] (11)--(6)--(13)--(12)--(8) (13)--(16) (2)--(6)--(4); \draw[-, line width=1pt] (13) circle (0.75);
}
\caption{A corner tree and its occurrence in a permutation}
\label{atree}
\end{figure}

\begin{definition*}
An \emph{occurrence} of a corner tree $T$ in a permutation $\pi \in S_n$ is a map 
$$i:\text{vertices}(T) \;\to\;\{1,\dots,n\}$$
such that the image of every vertex $v$ relative to its parent is compatible with its label in the following sense:
\begin{itemize}
\item[\texttt{*W }] 
$i(v) < i(v\mathrm{.parent})$
\item[\texttt{*E }] 
$i(v) > i(v\mathrm{.parent})$
\item[\texttt{S* }] 
$\pi(i(v)) < \pi(i(v\mathrm{.parent}))$ 
\item[\texttt{N* }] 
$\pi(i(v)) > \pi(i(v\mathrm{.parent}))$ 
\end{itemize}
\end{definition*}

\begin{remark*}
Here \texttt{*W} means that the label of $v$ is either \texttt{SW} or \texttt{NW}, and similarly in the other cases.
\end{remark*}

\begin{remark*}
In general, the map $i$ is not required to be injective.
\end{remark*}

\begin{definition*}
Let $\#T(\pi)$ be the \emph{count} of~$T$, which is the number of different occurrences of $T$ in~$\pi$. When $\pi$ is not important we write $\#T$. 
\end{definition*}

We give some examples of corner trees and their counts, showing that the latters are also pattern counts. This relation justifies our reuse of terminology.

\begin{example*}
Let $T=\;\ct{\root/1/0/10}\;$. Then $\#T(\pi)=|\pi|$, hence $\#T=\#\textit{1}\;$.
\end{example*}

\begin{example*}
Let $T'=\;\ct{\root/1/0/10,NE/2/0/10}\;$. Then $\#T'=\#\textit{12}\;$.
\end{example*}

Indeed, an occurrence of $T'$ is a map of its two vertices to two positions $i_1 < i_2$ such that $\pi(i_1) < \pi(i_2)$, which is an occurrence of the pattern \textit{12} in $\pi$. Similarly,
\begin{example*}
$\#\;\ct{\root/1/0/10,NE/2/0/10,NE/3/0/20}
\;\;=\;\;\#\textit{123} \;\;=\;\; \#\; \ct{SW/1/0/10,\root/2/0/10,NE/3/0/20} $
\end{example*}
Note that these two corner trees counting \emph{123} are essentially the same one with different choices of the root. The role of the root will become clear later, in the algorithm. Some corner trees yield combinations of pattern counts:
\begin{example*}
$ \#\;\ct{\root/1/0/10,SE/2/0/10,NE/3/0/20}
\;\;=\;\;\#\textit{213}\;+\;\#\textit{312} $
\end{example*}
Indeed, either \textit{213} or \textit{312} occur at $i_1 <i_2<i_3$ whenever \mbox{$\pi(i_1) > \pi(i_2) < \pi(i_3)$}. This kind of a combinations is called a \emph{partially ordered pattern} in~\cite{kitaev2005partially}, and was considered also in~\cite{hartung2020combinatorial}. Other corner trees lead to more involved linear combinations:
\begin{example*}
$\#\;\ct{\root/1/0/10,NE/2/0/10,SW/3/0/20}
\;\;=\;\;2\cdot\#\textit{123}\;+\;2\cdot\#\textit{213}\;+\;\#\textit{12}$
\end{example*}
Here we have considered all the ways the vertex images may be related, including the case that the root and the leaf \texttt{SW} map to the same entry. 

\begin{remark*}
It has been suggested by reviewers of this paper to replace the five vertex labels of a corner tree by the four labels $\{\swarrow, \searrow, \nwarrow, \nearrow\}$ on the \emph{edges}, and to name it a \emph{tree pattern}. This equivalent terminology may be useful in future studies, but here we stick to corner trees which are more suitable for our algorithm below. Another variation, hinted by the eight-wind compass rose label, is allowing also the vertex labels $\{$\texttt{N}, \texttt{S}, \texttt{E}, \texttt{W}$\}$. While our results extend to this setting, it would not yield new pattern statistics, only combinations of those given by corner trees.  
\end{remark*}

\medskip
In general, as we now show, the count of a corner tree is a linear combination of counts of patterns. We also give a formula to find this combination, which is particularly useful for our computer-generated results.

\begin{lemma}
\label{comb}
Let $T$ be a corner tree. Then there exists a finite list of integer coefficients $\{\widehat{\#}T(\sigma)\}_{\sigma \in \Sigma(T)}$ such that for every permutation $\pi \in \bigcup_{n=1}^{\infty} S_n$ 
$$ \#T(\pi) \;=\; \sum\limits_{\sigma \in \Sigma(T)} \widehat{\#}T(\sigma) \, \#{\sigma}(\pi) $$
Moreover, the coefficients $\widehat{\#}T(\sigma)$ are uniquely determined by the recursive formula
$$ \widehat{\#}T(\sigma) \;=\; \#T(\sigma) - \sum\limits_{|\rho| < |\sigma|} \widehat{\#}T(\rho) \, \#{\rho}(\sigma) $$
where the sum is over all patterns $\rho$ that are shorter than $\sigma$.
\end{lemma}

\begin{proof}
Let $\widehat{\#}T(\sigma)$ be the number of occurrences of $T$ in $\sigma$ which are onto the set of positions $\{1,\dots,|\sigma|\}$.

Every occurrence $i$ of a tree $T$ in a permutation $\pi$, defines a unique pattern $\sigma$ that occurs in $\pi$ at the image of $i$. In other words, we can uniquely represent the occurrence as $i = i'' \circ i'$ where $i'$ is an occurrence of $T$ onto~$\sigma$ and $i''$ is an occurrence of $\sigma$ in~$\pi$. Conversely, any such composition, of occurrences of $T$ onto~$\sigma$ and $\sigma$ in~$\pi$, gives rise to an occurrence of $T$ in $\pi$. The lemma now follows by dividing into cases according to $\sigma$.

Clearly $|\sigma| \leq |T|$, where $|\sigma|$ is the length of the pattern, and $|T|$ is the number of vertices in the tree. Thus the set of potential patterns $\Sigma(T)$ is finite and does not depend on $\pi$.

The formula for the coefficients follows by renaming the variables $\pi,\sigma$ as $\sigma,\rho$ and noting that for $|\rho|\geq|\sigma|$, we have $\#\rho(\sigma)=1$ if $\rho=\sigma$ and $0$ otherwise.
\end{proof}

In the other direction, counting a pattern \emph{may} reduce to counting corner trees. By the above examples, the pattern count \#\textit{12} is also a corner tree count, and same for \#\textit{123} and all increasing patterns. Since $2\,\#\textit{123}+2\,\#\textit{213}+\#\textit{12}$ is also a corner-tree count, one can solve for \#\textit{213} and express it as a combination of three corner-tree counts. 

\begin{remark*}
Lemma~\ref{comb} lets us work with finite combinations of pattern counts and tree counts, and derive relations that hold for any input of any length. To emphasize this point, we may identify formal sums of trees with formal sums of patterns. 

Let $V$ be the linear space of all formal linear combinations of patterns, that is, expressions of the form $A = \sum_{\sigma} a_{\sigma} \sigma$ where $a_{\sigma} \in \mathbb{Q}$. Denote the resulting permutation statistics $\#A = \sum_{\sigma} a_{\sigma} \#\sigma$. By Lemma~\ref{comb} every corner tree $T$ is an element of~$V$. A \emph{corner tree formula} is a combination of the form $\sum_T c_T T$, and we are interested in representing patterns by such formulas. Clearly, corner tree formulas form a subspace $U$ of~$V$. 

We remark without going into details that one can also define a multiplication that satisfies $\#(A \cdot B) = \#A\,\#B$, thereby giving both $U$ and $V$ a commutative ring structure. The unit would be the empty tree/pattern~$\phi$, with $\#\phi \equiv 1$, a constant permutation statistics. Such a definition of $A \cdot B \in U$ given $A,B \in U$ shows that sums of corner trees are not more restrictive than expressions that involve products of tree counts. Therefore, we do not pursue this idea here.  
\end{remark*}

\bigskip \noindent \textbf{Proof of Theorem~\ref{main}} \medskip

We now show that counting a corner tree in a permutation can be done in nearly linear time in its length. The pseudo-code to compute $\#T(\pi)$ is presented in Algorithm~\ref{alg}. The input to the algorithm is the permutation $\pi$, the corner tree~$T$, and a vertex~$v \in T$. 

If $v$ is the root, then the output will be~$\#T(\pi)$. Otherwise, for reasons explained below, the output will be an array of counts $C$, so that $C[x]$ is the number of valid maps of the subtree at $v$, given that $v$'s parent is sent to the position~$x$.

\begin{algorithm}[t]
\caption{Count a corner tree}
\linespread{1.25}\selectfont
\label{alg}
\begin{algorithmic}[1]
\Function{count-ct}{permutation $\pi$, tree $T$, vertex $v$}
\State $A \leftarrow \textsc{array}(1,\dots,1)$
\Comment{of size $n=|\pi|$}
\For{$w$ \textbf{in} $v$.children}
\State $A \leftarrow A \,\odot$ \textsc{count-ct}$(\pi, T, w)$
\Comment{term-wise product of arrays}
\EndFor
\If{$v\text{.label} = \root$}
\State \Return \textsc{sum}$(A)$
\Comment{of the $n$ values}
\EndIf
\State $C \leftarrow \textsc{array}(0,\dots,0)$
\Comment{size $n$}
\State $B \leftarrow \textsc{sum-tree}(0,\dots,0)$
\Comment{size $n$}
\For{$x$ \textbf{in} $\begin{cases} (1,\dots,n) & \text{if }v\text{.label} = \texttt{*W} \\ (n,\dots,1) & \text{if }v\text{.label} = \texttt{*E} \end{cases}$ \;\;\;}
\Comment{two cases}
\State $C[x] \leftarrow \begin{cases} B.\text{prefix-sum}(\pi[x]) & \text{if }v\text{.label} = \texttt{S*} \\ B.\text{suffix-sum}(\pi[x]) & \text{if }v\text{.label} = \texttt{N*} \end{cases} $
\Comment{two cases}
\State $B[\pi[x]] \leftarrow A[x] $
\EndFor
\State
\Return $C$
\Comment{the whole array}
\EndFunction
\end{algorithmic}
\end{algorithm}

Lines 2-4 produce an array $A$, such that $A[x]$ is the number of occurrences of the subtree rooted at $v$, that map $v$ to~$x$. This number is computed for every $x$ as a product over $v$'s children, as we have to choose valid maps for their subtrees. We thus proceed by recursion on the tree, and call the function once for every child. Note that if $v$ is a leaf then there is a single such map and $A[x]=1$ for every~$x$. If $v$ is the root, then we return the total count, on lines~5-6, as it does not matter where it maps.

Otherwise, we create and return an array of counts $C$ as above. Every entry of~$C$ is a sum of entries of $A$, such that their relation in $\pi$ is compatible with $v$'s label. For example, $C[x]$~will be the sum of $A[x']$ with $x'<x$ and $\pi(x')<\pi(x)$ in the case that $v$'s label is~\texttt{SW}. The other three cases are similar. To avoid a double loop and save time, we compute all these sums in one iteration over $A$, by inserting it into a sum-tree array~$B$, from which we properly retrieve $C[x]$. 

Note, for example, that this algorithm reduces to the above method for counting increasing patterns~$12 \dots k$, in the case that the tree~$T$ is a path of~\texttt{SW}.

In conclusion, Algorithm~\ref{alg} correctly counts corner trees. Overall, the function is called $k=|T|$ times. The arrays $A$, $B$, $C$ hold $O(n)$ numbers. On lines 4, 6, 9 there are loops of $n$ iterations. Thanks to sum trees, lines 10 and 11 are performed with at most $\log n$ arithmetic operations each. All the numbers involved require at most $k \log (n+1)$ bits. Thus for a $k$-vertex corner tree the function runs in $\tilde{O}(k^2n)$ time. It follows that any fixed permutation statistic expressible by corner trees can be computed in $\tilde{O}(n)$ time, and this proves Theorem~\ref{main}.

\section{Small Patterns}\label{small}

Before using corner trees to count 3-patterns and 4-patterns, we mention some well-known symmetries that reduce the number of cases. Let $\mathrm{rev}:S_n \to S_n$ be the reverse map, $\mathrm{rev}(\pi(1)\dots\pi(n)) = \pi(n)\dots\pi(1)$. Also, let $\mathrm{inv}(\pi) = \pi^{-1}$. These two transformations act as reflections on the graph of the permutation in the plane, and generate~$D_4$, the order-eight dihedral group of symmetries of the square. 

Clearly these eight symmetries respect pattern occurrence, in the sense that $\#[\mathrm{rev}\,\sigma](\mathrm{rev}\,\pi) = \#{\sigma}(\pi)$ for example. Thus we obtain linear time reductions between counting $\sigma$ and any other pattern in its orbit under the action of $D_4$. We may also apply these symmetries to corner trees, where $\textrm{rev}$ swaps \texttt{W} and~\texttt{E}, and $\textrm{inv}$ swaps \texttt{NW} and~\texttt{SE}. This yields direct reductions between corner tree formulas. 

For example, the problem of computing \#\emph{321} is equivalent to \#\emph{123}. Similarly, \#\emph{132}, \#\emph{231}, \#\emph{213} and \#\emph{312} are pairwise equivalent as rotations of each other. In Section~\ref{trees} we have seen that \#\emph{123} and \#\emph{213} are computable with corner trees. It follows that every 3-pattern can be counted in $\tilde{O}(n)$ time, as claimed by Corollary~\ref{three}. 

The two cases \emph{123} and \emph{213} extend to \emph{123...k} and \emph{213...k} for every~$k$. For example, we have the following formulas for counting 4-patterns.

\begin{align*}
\#\textit{1234} \;&=\; \;\#\; \ct{\root/1/0/10,NE/2/0/10,NE/3/0/20,NE/4/0/30} \\[-1.5em]
\#\textit{2134} \;&=\; \;-\; \#\textit{1234} \;-\; \tfrac12\; \#\textit{123} \;+\; \tfrac12 \;\#\; \ct{\root/1/1/11,SW/2/1/11,SW/3/1/21,SW/3/2/21} 
\end{align*}

\smallskip
We can give another corner tree formula for a 4-pattern, which is remarkably more complicated. Let
\begin{align*}
& \\[-1.5em]
A\;\;&=\;\;
\ct{\root/1/0/10,SW/2/0/10,SE/3/0/20,SE/2/1/10} \;\;-\; \ct{\root/1/0/10,SE/2/0/10,SW/3/0/20,SE/2/1/10} \;\;+\; \ct{\root/1/0/10,SE/2/0/10,NE/3/0/20,SE/2/1/10} \\[1.5em]
\;&-2\; \ct{\root/1/0/10,SW/2/0/10,SE/3/0/20,SE/4/0/30} \;\;-2\; \ct{\root/1/0/10,SW/2/0/10,SE/3/0/20,SW/4/0/30} \;\;-\; \ct{\root/1/0/10,SW/2/0/10,SE/3/0/20} \\[1em]
\;&+\frac13\; \ct{\root/1/0/10,NE/2/0/10,NE/1/1/10,NE/2/1/10} \;\;-\; \ct{\root/1/0/10,NE/2/0/10,NE/3/0/20,NE/3/1/20} \;\;-\frac12\; \ct{\root/1/0/10,NE/2/0/10,NE/2/1/10} \;\;+\frac16\; \ct{\root/1/0/10,NE/2/0/10} \\[-1em] &
\end{align*}

\begin{lemma}
$\#A \;=\; \#\textit{2143}$
\end{lemma}

\noindent
\emph{Proof.}
This follows by a routine inspection of pattern occurrences counted by these trees. Denote the trees in $A$ by $T_1,\dots,T_{10}$ in order of appearance. The desired pattern \emph{2143} is counted once by $T_1$ and nowhere else. Every other pattern is cancelled out when summing all the terms. We include a brief account for the interested reader.

\begin{enumerate}
\item[(a)]
\emph{2413, 2431} are also counted by $T_1$, but subtracted back by~$T_2$.
\item[(b)]
\emph{4123, 4132, 4213, 4231, 4312, 312} only appear in $T_2$ and $T_3$, the same number of times with opposite signs.
\item[(c)] 
\emph{3124, 3142, 3214, 3241, 3412, 3421} are counted twice among $T_1$ and $T_3$, and \emph{231, 213} counted once. They are cancelled out by $T_4$, $T_5$ and~$T_6$.
\item[(d)] 
\emph{1324, 1342, 1423, 1432, 132} are added twice by $T_7$ to compensate for their unnecessary subtraction by $T_2$ or~$T_5$.
\item[(e)] 
\emph{1234, 1243, 123, 132, 12} appear in~$T_7$ and are balanced by $T_8$, $T_9$, $T_{10}$.\qed
\end{enumerate}

In conclusion, we have presented nearly linear time algorithms that count all the patterns from three out of the seven $D_4$-orbits in $S_4$: $\{\textit{1234}$, $\textit{4321}\}$,  $\{\textit{1243}$, $\textit{2134}$, $\textit{3421}$, $\textit{4312}\}$,  $\{\textit{2143}$, $\textit{3412}\}$. This proves Corollary~\ref{the-eight}.

\begin{remark*}
It was noted by a reviewer that, by known properties of these three orbits \cite[A005802]{oeis}, all the eight resulting algorithms return zero on the same number of order-$n$ inputs, for every~$n$.
\end{remark*}

\begin{remark*}
One way to describe the eight patterns in Corollary~\ref{the-eight} is as monotone $\textit{1234}$ or $\textit{4321}$, with possible adjacent transpositions at the two ends. The same procedure gives eight patterns for any~$k \geq 4$. Letting $k=6$ for example, these patterns are: 
$$ \textit{123456},\; \textit{123465},\; \textit{213456},\; \textit{213465},\; \textit{564312},\; \textit{564321},\; \textit{654312},\; \textit{654321} $$
We note that these patterns are counted by corner trees as well. For $k \geq 5$ this is done by a similar technique to $\textit{1234}$ and $\textit{2134}$ above. 
\end{remark*}

\bigskip \noindent \textbf{The Bergsma--Dassios-Yanagimoto Statistic} \medskip

Let $(X_1,Y_1),\dots,(X_k,Y_k)$ be independent samples in the real plane, drawn from a common bivariate distribution $(X,Y)$ with non-atomic marginals. Recall that such samples induce, with probability one, a random permutation $\rho \in S_k$ such that $\textrm{rank}\,Y_1 = \rho(\textrm{rank}\,X_1),\;\textrm{rank}\,Y_2 = \rho(\textrm{rank}\,X_2),\;\dots,\;\textrm{rank}\,Y_k = \rho(\textrm{rank}\,X_k)$. Equivalently, the samples can be written as $(X_{(1)},Y_{(\rho(1))}),\dots,(X_{(k)},Y_{(\rho(k))})$ with the usual order statistics notation, $X_{(1)} < X_{(2)} < \dots < X_{(k)}$. As noted above in the introduction, many nonparametric measures of bivariate data can be described in terms of this random permutation.

In a typical statistical test, the null hypothesis states that the two variables $X$ and $Y$ are independent. Under this assumption, it is easy to see that for all $k \in \mathbb{N}$, 
$$ (\star) \;\;\;\;\;\;\;\; \;\;\;\;\;\;\;\; \forall \; \sigma \in S_k \;\;\;\;\;\;\;\; P(\rho = \sigma) \;=\; \tfrac{1}{k!} \;\;\;\;\;\;\;\; $$
A classical result by Hoeffding gives a strong converse~\cite{hoeffding1948non}. If $(\star)$ holds for some $k \geq 5$ with $X$ and $Y$ jointly absolutely continuous, then they are independent. Indeed, the well-known nonparametric independence test by Hoeffding, which can be expressed with counts of 5-pattern, is consistent against all alternative joint densities.

Hoeffding's test has been extended in later works to variables in $\mathbb{R}^d$, to three or more random variables, and in other directions. Some variants that relax the continuity assumption have been studied. Yanagimoto has shown that independence actually follows from $(\star)$ with $k=4$, and this is optimal since $(\star)$ with $k=3$ is not sufficient~\cite{yanagimoto1970measures}. Bergsma and Dassios have suggested a corresponding nonparametric independence test, based on 4-patterns, which has gained popularity in recent years~\cite{bergsma2010nonparametric,bergsma2014consistent}. Up to normalization, the Bergsma--Dassios--Yanagimoto statistic has the following simple form:
$$ T^{\star} \;=\; \#\textit{1234} + \#\textit{1243} + \#\textit{2134} + \#\textit{2143} + \#\textit{3412} + \#\textit{3421} + \#\textit{4312} + \#\textit{4321} $$
Under independence, $T^{\star}/\tbinom{n}{4}$ is concentrated at $1/3$, while for any other bivariate distribution with continuous marginals these eight patterns appear at a strictly higher overall rate~\cite{yanagimoto1970measures, bergsma2010consistent, drton2018high}. Natural interpretations of this statistical measure have been discussed in recent works~\cite{bergsma2014consistent,weihs2018symmetric,even2018patterns}.

We remark that the relation between independence and pattern frequencies has been independently discovered in the theory of so-called \emph{permutons}, limit objects for permutations. In that context, it means that quasi-randomness of permutations is ``finitely forcible''~\cite{kral2013quasirandom, glebov2015finitely}.

\medskip
A straightforward evaluation of $T^{\star}$ requires $\tilde{O}(n^4)$ time, though it was noted right away that this can be reduced to $\tilde{O}(n^3)$~\cite{bergsma2010consistent}. Therefore, it was suggested at first to approximate this statistic by a sufficiently large random sample of subsets of four entries~\cite[Sect.~4]{bergsma2014consistent}. This method is applicable, but due to the degenerate nature of the null distribution~\cite{nandy2016large,dhar2016study,even2018patterns} it incurs a substantial loss of information even if as much as $O(n^2)$ random samples are used~\cite{janson1984asymptotic}. Two recent papers have been dedicated to an exact computation of~$T^{\star}$ in $\tilde{O}(n^2)$ time~\cite{weihs2016efficient, heller2016computing}. These methods might still be infeasible in big data applications, where only near linear-time work is considered practical.

Corollary~\ref{bd}, stating that indeed the Bergsma--Dassios--Yanagimoto statistic can be evaluated in $\tilde{O}(n)$ time, is immediate from Corollary~\ref{the-eight}. Nevertheless, we derive a more direct formula, based on the invariance of $T^{\star}$ under the symmetries mentioned above. Let the action of the dihedral group linearly extend from patterns and trees to their formal sums, so that $\#S(\pi) = \#[g\cdot S](g \cdot \pi)$ for a corner tree formula $S$ and a reflection or a rotation $g \in D_4$.

\begin{proposition}
\label{bdsym}
The Bergsma--Dassios--Yanagimoto permutation statistic is obtained from averaging the corner tree formula,
$$ S\;\;=\;\;2\;\;  \ct{\root/2/2/22,\emph{SE}/3/2/22,\emph{NE}/2/3/22,\emph{NE}/3/3/22} \;\;+2\;\; \ct{\root/2/2/22,\emph{NE}/2/3/22,\emph{SE}/3/2/23,\emph{NE}/3/3/32} \;\;-2\;\; \ct{\emph{NE}/2/2/22,\root/2/1/22,\emph{SE}/3/1/21,\emph{NE}/3/2/31} \;\;-\;\; \ct{\root/2/2/22,\emph{SE}/3/2/22,\emph{NE}/3/3/22}$$
over all eight reflections and rotations, as follows:
$$ T^{\star}(\pi) \;=\; \binom{|\pi|}{4} - \frac18 \sum\limits_{g \in D_4} \#[g \cdot S](\pi) $$
\end{proposition}

\begin{proof}
This is a straightforward calculation. In spirit of the discussion following Lemma~\ref{comb}, expand $S$ in patterns:
\begin{align*}
S \;&=\; 2(2\cdot\textit{2134} + 2\cdot\textit{2143} + 2\cdot\textit{2314} + 2\cdot\textit{2341} + 2\cdot\textit{2413} + 2\cdot\textit{2431} + \textit{213} + \textit{231}) \\
&+\; 2(\textit{1324} + \textit{1423} + \textit{2314} + \textit{2413} + \textit{3412}) \\
&-\; 2(2\cdot\textit{2134} + 2\cdot\textit{2143} + \textit{2314} +  \textit{2413} + \textit{3124} + \textit{3142} + \textit{3412} + \textit{213}) \\
&-\; (\textit{213}+\textit{231})
\end{align*}
Sort the patterns into their $D_4$ orbits:
\begin{align*}
S&\;=\; (\textit{231}-\textit{213}) \\
&+\; 2(\textit{1324}) + 4(\textit{2341}) + 2(2\cdot\textit{2413}-\textit{3142}) + 8(\tfrac14\textit{1423}+\tfrac12\textit{2314}+\tfrac12\textit{2431}-\tfrac14\textit{3124})
\end{align*}
Let $S^{\star} = \frac18\sum_{g\in D_4}[g \cdot S]$. Then $(\textit{231}-\textit{213})$ vanishes, and the other orbits equalize: 
\begin{align*}
S^{\star} \;&=\; (\textit{1324}+\textit{4231}) \;+\; (\textit{1432}+\textit{2341} + \textit{3214} + \textit{4123}) \;+\; (\textit{2413} + \textit{3142}) \\
\;&+\; (\textit{1342} + \textit{1423} + \textit{2314} + \textit{2431} +\textit{3124} + \textit{3241}+\textit{4132} + \textit{4213})
\end{align*} 
The proposition now follows by the definition of $T^\star$, noting that the total count of all twenty-four order-$4$ patterns in an order-$n$ permutation is $\tbinom{n}{4}$. 
\end{proof}

\begin{remark*}
The symmetrized corner tree formula $S^{\star}$ has 20 nonequivalent terms. To simplify the computation, one can equivalently write $\#S(g^{-1} \cdot \pi)$ in the proposition instead of $\#[g \cdot S](\pi)$, and so always deal with the four corner trees in~$S$. Note that these trees have some overlapping branches, which may be exploited by a dedicated algorithm in practical computations. We provide implementations of these different approaches in the python package, see details below.
\end{remark*}


\bigskip \noindent \textbf{Corner Tree Spaces} \medskip

As noted after Lemma~\ref{comb}, formal sums of corner trees constitute a subspace of formal sums of patterns. This means that certain linear combinations of pattern counts may be computable by corner trees, even if individual patterns in those combinations are not so. 

For example, we have seen eight order-4 patterns, out of 24, that are contained in the subspace spanned by order-4 corner trees. However, a computer-assisted calculation shows that this subspace is actually 23-dimensional, proving Proposition~\ref{23}. This means that the 4-profile can be computed in near linear time up to only one freedom degree, an unknown multiple of one particular vector. 

What is the missing element? Endow the space of formal sums of 4-patterns with the unique inner product that makes the single 4-patterns an orthonormal basis. One can solve for a combination perpendicular to the pattern expansions of corner trees as above, and obtain:
\begin{align*}
N \;&=\; (\textit{1324}+\textit{4231}) \;+\; (\textit{1432}+\textit{2341} + \textit{3214} + \textit{4123}) \;+\; (\textit{2413} + \textit{3142}) \\[0.2em]
&-\; (\textit{1342} + \textit{1423} + \textit{2314} + \textit{2431} +\textit{3124} + \textit{3241}+\textit{4132} + \textit{4213})
\end{align*} 
Therefore, the subspace of linear combinations of 4-patterns for which we have corner tree formulas is characterized as~$N^{\perp}$, with respect to the above inner product.

\begin{remark*}
Here is an application of these computations. In the setting of Hoeffding's test, one may look for other independence-detecting sums $F=$ \raisebox{0.1em}{$\sum_{\sigma \in S_k}$}$f_{\sigma} \#\sigma$, where $E[F/\tbinom{n}{k}] = \alpha$ assuming independence, and $E[F/\tbinom{n}{k}] > \alpha$ under all dependent alternatives. Such combinations form a convex cone, studied for 4-patterns in a recent paper~\cite{chan2019characterization}. Ten so-called $\Sigma$-forcing combinations with such a property are provided there, including the Bergsma--Dassios--Yanagimoto statistic. We note that six out of the ten are contained in the subspace of corner trees, and hence can be computed efficiently.
\end{remark*}

We conclude our discussion of corner tree formulas with some brief details on their implementation. The Python package \texttt{cornertree} accompanies this article~\cite{even2019corners}. It includes the general-purpose classes \texttt{SumTree}, \texttt{CornerTree}, and \texttt{CornerTreeFormula}, as well as an implementation of Algorithm~\ref{alg}, with a plenty of examples. The function \texttt{Tstar} performs an optimized computation of the Bergsma--Dassios--Yanagimoto statistic. The iterator \texttt{canonical\_corner\_trees} generates all corner trees of a given order, and the routine \texttt{expand\_corner\_trees} expands them as formal sums in pattern. Proposition~\ref{23} and the other observations in the previous paragraphs follow from the output by standard linear algebra.

\medskip
The same procedure shows that corner trees span a 100-dimensional subspace of 5-patterns, a 463-dimensional subspace of 6-patterns, and a 2323-dimensional subspace of 7-patterns. The latter computation required several hours of CPU time.

\section{Faster 4-Patterns}\label{four}

We turn to explore other methods for counting patterns, beyond those provided by corner trees. This section focuses on the intriguing problem of computing the full 4-profile. Proposition~\ref{23} asserts that corner trees yield 23 independent linear functionals of this 24-dimensional vector. Therefore, counting any 4-pattern that is not already spanned by corner trees will accomplish the computation.

Here is a simple $\tilde{O}(n^2)$ solution, achieved by counting \emph{3214} in a given $\pi \in S_n$. Let $i_4 \in \{1,\dots,n\}$ be a position that will correspond to \emph{4}, the last entry in an occurrence of \emph{3214}. We count $i_1,i_2,i_3$ that complete an occurrence. This is done by restricting $\pi$ to smaller positions $x<i_4$ with smaller values~$\pi(x)<\pi(i_4)$, and counting \emph{321} exactly as described in the beginning of Section~\ref{trees}. To obtain \#\emph{3214} we sum over $i_4$, adding up $n$ numbers, each computed in $\tilde{O}(n)$ time.

\bigskip \noindent \textbf{Faster \#\emph{3214}}
\medskip

The above method gives rise to much overlapping work, as the same occurrences of \emph{321} are counted over and over, for different values of $i_4$. Our approach is to combine these tasks together for adjacent $i_4$ and adjacent $\pi(i_4)$. 

We thus divide $\{1,\dots,n\}$ into chunks of size $m$, where $m$ is a parameter to be suitably chosen later. We define two properties for an occurrence of~\emph{3214} at positions~$i_1,i_2,i_3,i_4$:
\begin{itemize}
\item[(A)]
$\pi(i_1) \leq am < \pi(i_4)$ for some integer $a$.
\item[(B)]
$i_3 \leq bm < i_4$ for some integer $b$.
\end{itemize}
To describe the properties geometrically, divide the graph of $\pi$ in the plane into horizontal and vertical strips of width~$m$. Property~A asserts that the \emph{3} and the \emph{4} of the \emph{3214} are contained in different horizontal strips, and property~B similarly asserts that the \emph{1} and the \emph{4} are in different vertical strips.

We first count type~A occurrences in $\tilde{O}(n^2/m)$ time, as follows. For every integer $a < n/m$, we scan $\pi$ once, for the purpose of counting occurrences with \emph{4} in the $a$-th horizontal strip. Our action at position~$x$ depends on the value of~$\pi(x)$. If $\pi(x) > (a+1)m$ then we do nothing. If $\pi(x) \leq am$ then we use it to count $\emph{321}$ with two sum-trees as above. If $\pi(x) \in [am+1,(a+1)m]$, then we add the number of \emph{321} counted so far to the total count of~\emph{3214}.

By applying the same method to the transposed permutation $\pi^{-1}$, we also count occurrences of \emph{3214} that satisfy~B. 

Occurrences that satisfy both A and~B have now been counted twice, and need to be subtracted. We hence make the following adjustments to the counting of type~A. Whenever $i$ reaches a multiple of $m$, we save the count of \emph{321} at that time. Upon observing $\pi(i)$ in the $a$-th strip, we add the current count of \emph{321} as before, but subtract the saved count from the previous $m$-multiple. Thereby, we subtract all occurrences of type~A$\cap$B, compensating for their double counting.

It is left to deal with occurrences of \emph{3214} satisfying neither A nor~B. Given any position $i_4$, there are at most $m$ relevant $i_3 < i_4$ in the same vertical strip, and at most $m$ relevant $i_1$ with $\pi(i_1)<\pi(i_4)$ in the same horizontal strip. We iterate through all $O(nm^2)$ such $(i_1,i_3,i_4)$, and count $i_2$ that complete an occurrence of~\emph{3214}. This count is the number of points $(x,\pi(x))$ in the box $(i_1,i_3)\times(\pi(i_3),\pi(i_1))$. Similar to sum-trees, there exist a simple data structure that efficiently answers such queries, which we turn to describe here in brief.

An $n$-by-$n$ \emph{two-dimensional sum-tree} uses a Cartesian product of two complete binary trees of depth $\lceil \log n \rceil$. Every ordered pair of a vertex from each tree holds a number. The pairs of leafs $(l,l')$ thus hold an $n \times n$ array of numbers. A~general pair of vertices $(v,v')$ always holds the sum over pairs of leafs from the corresponding subtrees. Insertion makes access to $O(\log^2n)$ pairs, as well as querying a box sum. If the array is sparse then we may use a hash table, and only store the $O(s\log^2n)$ nonzero entries, where $s$ is the number of nonzero array cells. See~\cite{bentley1975multidimensional,matouvsek1994geometric,deBerg2008,chan2011orthogonal} for various general data structures of this kind. 

Using a two-dimensional sum-tree, with ones at all $(x,\pi(x))$ and zeros elsewhere, the remaining occurrences are counted in $\tilde{O}(nm^2)$ time and $\tilde{O}(n)$ space, as explained above. In order to balance $n^2/m$ and $nm^2$, we set $m = \lfloor n^{1/3} \rfloor$, so that the total running time is~$\tilde{O}(n^{5/3})$. This algorithm together with Proposition~\ref{23} prove Theorem~\ref{n53}.

\bigskip \noindent \textbf{Faster \#\emph{3241}}
\medskip

Theorem~\ref{n32} is proven by another algorithm based on similar ideas, which computes \#\emph{3241} in $\tilde{O}(n^{3/2})$ time and space. As before we divide the occurrences $i_1<i_2<i_3<i_4$ of \emph{3241} into cases. This time we ask whether~\emph{2,1} and whether~\emph{3,4} are contained in the same horizontal strips of height~$m$.
\begin{itemize}
\item[(A)]
$\pi(i_4) \leq am < \pi(i_2)$ for some integer $a$.
\item[(B)]
$\pi(i_1) \leq bm < \pi(i_3)$ for some integer $b$.
\end{itemize}

Since this algorithm is somewhat more involved, we present it in full detail, see Algorithm~\ref{alg3241}. In the following paragraphs we outline the key steps.

\begin{algorithm}[bt]
\caption{Compute \#\emph{3241}}
\linespread{1.25}\selectfont
\label{alg3241}
\begin{algorithmic}[1]
\Function{count-3241}{permutation $\pi$ of size $n$, integer $m$}
\State $t \leftarrow 0$
\Comment{accumulating $\#\textit{3241}$}
\For{$chunk$ \textbf{in} $\{1\dots m\},\{m+1 \dots 2m\} \dots \{\lfloor \tfrac{n}{m} \rfloor m + 1 \dots n\}$}
\State $c \leftarrow 0$
\Comment{$\#\textit{3241}$ given \emph{1} or \emph{4} in $chunk$}
\State $A_1, A_{12}, B_1, B_{21}  \leftarrow 4 \times \textsc{sum-tree}(0 ,\dots, 0)$
\Comment{size $n$ each}
\For{$y$ \textbf{in} $\{\pi(1),\pi(2) ,\dots, \pi(n)\}$}
\If{$y < \min(chunk)$}
\Comment{below strip}
\State{$B_1[y] \leftarrow 1$}
\State{$B_{21}[y] \leftarrow B_1[>y]$}
\Comment{means suffix-sum}
\State{$c \;\leftarrow\; c + B_1[>y] \cdot \left(y-1-B_1[<y]\right) - B_{21}[>y]$}
\Comment{type $B$}
\EndIf
\If{$y > \max(chunk)$}
\Comment{above strip}
\State{$A_1[y] \leftarrow 1$}
\State{$A_{12}[y] \leftarrow A_1[<y]$}
\Comment{means prefix-sum}
\State{$c \;\leftarrow\; c + \binom{A_1[<y]}{2} - A_{12}[<y]$}
\Comment{type $A$}
\State{$y' \;\leftarrow\; \left\lfloor\tfrac{y-1}{m}\right\rfloor m$}
\Comment{$m$-multiple below $y$}
\State{$c \;\leftarrow\; c - \binom{A_1\left[\leq y'\right]}{2} + A_{12}\left[\leq y'\right]$}
\Comment{type $A \cap B$}
\EndIf
\If{$y \in chunk$} {$t \leftarrow t + c$}
\EndIf
\EndFor
\EndFor
\State $P \leftarrow \textsc{product-tree}(n \times n)$
\Comment{2-dim array, all 0}
\For{$y$ \textbf{in} $\left\{1,2,\dots, n\right\}$}
\State{$y' \;\leftarrow\; \left\lfloor\tfrac{y-1}{m}\right\rfloor m$}
\Comment{$m$-multiple below $y$}
\For{$z$ \textbf{in} $\left\{y'+1,y'+2 ,\dots, y-1\right\}$}
\State{$P[\pi^{-1}(z),\pi^{-1}(y)] \leftarrow 1$}
\Comment{product-tree insertion}
\EndFor
\For{$z$ \textbf{in} $\left\{y+1,y+2 ,\dots, y'+m\right\}$}
\If{$\pi^{-1}(y)<\pi^{-1}(z)$} 
\State{$t \leftarrow t + P\textrm{.sum-box}\; \left(\pi^{-1}(z),n\right] \times \left(\pi^{-1}(y),\pi^{-1}(z)\right)$}
\Comment{query}
\EndIf
\EndFor
\EndFor
\State \Return $t$
\EndFunction
\end{algorithmic}
\end{algorithm}

The first part of the algorithm counts occurrences of types A and~B. We thus iterate over horizontal strips as before and count \emph{3241} with: (a) the~\emph{1} in the strip and the \emph{324} above it; (b) the~\emph{4} in the strip and the \emph{321} below it.

When counting type~B, we have to make sure that the~\emph{4} falls between the~\emph{2} and the~\emph{1}. Therefore, as we go over the entries of the permutation, we increment the count of \emph{321} upon observing the~\emph{2} and decrease it back upon observing the~\emph{1}. See the two corresponding terms on line~10 of the algorithm.

As for type~$A$, we count \emph{324} by adding \emph{234+324} and subtracting \emph{234}. This is similar to the treatment of \#\emph{213} in Section~\ref{trees}. See two terms on line~14. The subsequent lines~15-16 further subtract those type~A occurrences that are already counted as type~B for another strip.

The second part counts those \emph{3241} that are neither A nor~B, so that both the~\emph{21} and the~\emph{34} are contained in horizontal strips. We use a two-dimensional sum tree again, this time keeping account of all same-strip pairs of entries. Their coordinates in the tree are the position of the~\emph{1} and the position of the~\emph{2}.

We scan the graph of~$\pi$ bottom to top, inserting around $\frac{n}{m}\binom{m}{2}$ such pairs to the two-dimensional sum tree on line~22. In parallel, for every same-strip increasing pair \emph{34}, we lookup those \emph{21} below it that properly merge to~\emph{3241}. This number is efficiently retrieved by a suitable box sum query, on line~25.

The two parts of the algorithm run in time $\tilde{O}(n^2/m)$ and $\tilde{O}(nm)$. We hence set $m = \lfloor\sqrt{n}\rfloor$ to obtain $\tilde{O}(n^{3/2})$ time complexity. Since we make at most $nm$ insertions into the two-dimensional sum tree, it is sparse and only needs $\tilde{O}(n^{3/2})$ memory space. More generally, we obtain a tradeoff between $\tilde{O}(n^{2-a})$ computation time and $\tilde{O}(n^{1+a})$ space, where $a$ ranges from $0$ to $1/2$. 

We note that this algorithm still seems to contain much overlapping work. For example, the same \emph{321} and \emph{324} are counted many times as we iterate over horizontal strips. Further improvements are yet to be sought, perhaps by dividing into more cases.

\medskip
See also the python implementation, \texttt{count\_3214} and \texttt{count\_3241}, in our package \texttt{cornertree.py}~\cite{even2019corners}. It includes a general-purpose class \texttt{ProductTree} implementing two-dimensional sum-trees.

\section{Discussion}\label{discuss}

The following open problems and directions for future research naturally arise from the results presented in this paper.

\begin{enumerate}
\setlength\itemsep{0.5em}
\item
Are there pattern counts, or linear combinations of them, that are not spanned by corner trees, but still can be computed in $\tilde{O}(n)$ time?
\item
What is the computational complexity of finding the full 4-profile of a given permutation of size $n$?
\item
Is there a simple characterization of the linear space of pattern combinations that are spanned by corner trees?
\item
In particular, what is the dimension of the restriction of that space to patterns of size $k$?
Is it spanned by corner trees with $k$ vertices? 
\item
What other patterns can be counted, say in $\tilde{O}(n^{3/2})$ time, using the ideas presented in Section~\ref{four}?
\end{enumerate}

In the spirit of the last question, it would be interesting to initiate a systematic study, similar to our treatment of corner trees, of other tree-like structures and algorithmic ideas, and the linear spaces of patterns that they yield.

This may also be relevant to the remaining gap in the asymptotic complexity of $k$-pattern counting. Given $k \in \mathbb{N}$, what is the smallest $d(k)$ such that the $k$-profile can be computed in $\tilde{O}(n^{d(k)})$ time as $n$ grows? For example $d(2)=d(3)=1$ and $d(4) \leq \tfrac32$ by our results.  Recall that Berendsohn et al.~showed that for~$k$ large $d(k) \leq \frac{k}{4}+o(k)$, while $o(k/\log k)$ is not likely~\cite{berendsohn2019finding}.

We have computed the space of pattern combinations that are spanned by $k$-vertex corner trees for every $k \leq 7$. Its dimension is clearly bounded by the number of $k$-vertex corner trees, which is $O(11.07^k)$~\cite[A052763]{oeis}. Since the $k$-profile is $k!$-dimensional, a  super-exponential family of such objects would be necessary to span all $k$-patterns.

\medskip
As for more practical research directions, one may improve the computation time of other rank-based statistical tests, applying or extending our methods. Some potential examples are the general families of independence tests recently proposed in~\cite{heller2012consistent,heller2016consistent,weihs2018symmetric}.

By working with permutations, we always assume distinct inputs and distinct outputs, but the applications to data analysis also make sense for distributions with atoms, where ties occur with positive probability. It may hence be useful to extend our methods to such data sets, treating ties appropriately.

\medskip
Another potential extension concerns other families of permutation statistics studied in Combinatorics. In particular, can one adapt corner trees to count so-called generalized permutation patterns, such as vincular patterns~\cite{babson2000generalized}, partially ordered generalized patterns~\cite{kitaev2005partially} and mesh patterns~\cite{branden2011mesh}?

\medskip
Finally, we remark that while this work focuses on the problem of exact counting of patterns, also questions of approximation seem interesting and useful. Several works have been recently devoted to property testing of permutations~\cite{hoppen2011testing,klimovsova2014hereditary,fox2018fast,ben2018improved,newman2019testing}, which may be relevant to the study of approximated pattern counting. 

\subsection*{Acknowledgements}

The authors would like to thank the anonymous
reviewers for their helpful comments, questions and suggestions.

Chaim Even-Zohar was supported by the Lloyd’s Register Foundation / Alan Turing Institute programme on Data-Centric Engineering.

\bibliographystyle{alpha}
\bibliography{counting}

\newcommand{\etalchar}[1]{$^{#1}$}
\begin{thebibliography}{HHMW20}

\bibitem[AAAH01]{albert2001algorithms}
Michael~H Albert, Robert~EL Aldred, Mike~D Atkinson, and Derek~A Holton.
\newblock Algorithms for pattern involvement in permutations.
\newblock In {\em International Symposium on Algorithms and Computation}, pages
  355--367. Springer, 2001.

\bibitem[AAH{\etalchar{+}}02]{albert2002packing}
Michael~H Albert, Mike~D Atkinson, Chris~C Handley, Derek~A Holton, and Walter
  Stromquist.
\newblock On packing densities of permutations.
\newblock {\em Electron. J. Combin}, 9(1), 2002.

\bibitem[AR08]{ahal2008complexity}
Shlomo Ahal and Yuri Rabinovich.
\newblock On complexity of the subpattern problem.
\newblock {\em SIAM Journal on Discrete Mathematics}, 22(2):629--649, 2008.

\bibitem[BBL98]{bose1998pattern}
Prosenjit Bose, Jonathan~F Buss, and Anna Lubiw.
\newblock Pattern matching for permutations.
\newblock {\em Information Processing Letters}, 65(5):277--283, 1998.

\bibitem[BC11]{branden2011mesh}
Petter Br{\"a}nd{\'e}n and Anders Claesson.
\newblock Mesh patterns and the expansion of permutation statistics as sums of
  permutation patterns.
\newblock {\em The Electronic Journal of Combinatorics}, 18(2):5, 2011.

\bibitem[BC18]{ben2018improved}
Omri Ben{-Eliezer} and Cl{\'e}ment~L Canonne.
\newblock Improved bounds for testing forbidden order patterns.
\newblock In {\em Proceedings of the Twenty-Ninth Annual ACM-SIAM Symposium on
  Discrete Algorithms}, pages 2093--2112. SIAM, 2018.

\bibitem[BCKM15]{bliznets2015kernelization}
Ivan Bliznets, Marek Cygan, Pawe{\l} Komosa, and Luk{\'a}{\v{s}} Mach.
\newblock Kernelization lower bound for permutation pattern matching.
\newblock {\em Information Processing Letters}, 115(5):527--531, 2015.

\bibitem[BD10]{bergsma2010consistent}
Wicher Bergsma and Angelos Dassios.
\newblock A consistent test of independence based on a sign covariance related
  to {K}endall's tau.
\newblock {\em arXiv preprint arXiv:1007.4259}, 2010.

\bibitem[BD14]{bergsma2014consistent}
Wicher Bergsma and Angelos Dassios.
\newblock A consistent test of independence based on a sign covariance related
  to {K}endall's tau.
\newblock {\em Bernoulli}, 20(2):1006--1028, 2014.

\bibitem[Ben75]{bentley1975multidimensional}
Jon~Louis Bentley.
\newblock Multidimensional binary search trees used for associative searching.
\newblock {\em Communications of the ACM}, 18(9):509--517, 1975.

\bibitem[Ber10]{bergsma2010nonparametric}
Wicher Bergsma.
\newblock Nonparametric testing of conditional independence by means of the
  partial copula.
\newblock {\em arXiv preprint, arXiv:1101.4607}, 2010.

\bibitem[Ber19]{berendsohn2019complexity}
Benjamin~Aram Berendsohn.
\newblock Complexity of permutation pattern matching.
\newblock {\em Master Thesis, FU, Berlin}, 2019.

\bibitem[BH10]{burstein2010packing}
Alexander Burstein and Peter H{\"a}st{\"o}.
\newblock Packing sets of patterns.
\newblock {\em European Journal of Combinatorics}, 31(1):241--253, 2010.

\bibitem[BHL{\etalchar{+}}15]{balogh2015minimum}
J{\'o}zsef Balogh, Ping Hu, Bernard Lidick{\`y}, Oleg Pikhurko, Bal{\'a}zs
  Udvari, and Jan Volec.
\newblock Minimum number of monotone subsequences of length 4 in permutations.
\newblock {\em Combinatorics, Probability and Computing}, 24(04):658--679,
  2015.

\bibitem[BKM19]{berendsohn2019finding}
Benjamin~Aram Berendsohn, L{\'a}szl{\'o} Kozma, and D{\'a}niel Marx.
\newblock Finding and counting permutations via {CSP}s.
\newblock {\em arXiv preprint arXiv:1908.04673}, 2019.

\bibitem[BKR61]{blum1961distribution}
Julius~Robin Blum, Jack~Carl Kiefer, and Murray Rosenblatt.
\newblock Distribution free tests of independence based on the sample
  distribution function.
\newblock {\em The Annals of Mathematical Statistics}, pages 485--498, 1961.

\bibitem[BL12]{bruner2012fast}
Marie-Louise Bruner and Martin Lackner.
\newblock A fast algorithm for permutation pattern matching based on
  alternating runs.
\newblock In {\em Scandinavian Workshop on Algorithm Theory}, pages 261--270.
  Springer, 2012.

\bibitem[BL13]{bruner2013computational}
Marie-Louise Bruner and Martin Lackner.
\newblock The computational landscape of permutation patterns.
\newblock {\em arXiv preprint arXiv:1301.0340}, 2013.

\bibitem[BL16]{bruner2016fast}
Marie-Louise Bruner and Martin Lackner.
\newblock A fast algorithm for permutation pattern matching based on
  alternating runs.
\newblock {\em Algorithmica}, 75(1):84--117, 2016.

\bibitem[B{\'o}n07]{bona2007copies}
Mikl{\'o}s B{\'o}na.
\newblock The copies of any permutation pattern are asymptotically normal.
\newblock {\em arXiv:0712.2792}, 2007.

\bibitem[B{\'o}n10]{bona2010three}
Mikl{\'o}s B{\'o}na.
\newblock On three different notions of monotone subsequences.
\newblock {\em Permutation Patterns}, 376:89--114, 2010.

\bibitem[B{\'o}n12]{bona2012combinatorics}
Mikl{\'o}s B{\'o}na.
\newblock {\em Combinatorics of Permutations}.
\newblock Chapman and Hall / CRC Press, New York, US, 2012.

\bibitem[BS00]{babson2000generalized}
Eric Babson and Einar Steingr{\'\i}msson.
\newblock Generalized permutation patterns and a classification of the
  {M}ahonian statistics.
\newblock {\em S{\'e}minaire Lotharingien de Combinatoire [electronic only]},
  44:1--18, 2000.

\bibitem[Cha18]{chaturvedi2018geeks}
Aayush Chaturvedi.
\newblock Count inversions of size k in a given array.
\newblock {\em Geeks for Geeks, \texttt{www.geeksforgeeks.org}}, 2018.

\bibitem[CK16]{cooper2016complexity}
Joshua Cooper and Anna Kirkpatrick.
\newblock The complexity of counting poset and permutation patterns.
\newblock {\em Australasian Journal of Combinatorics}, 64(1):154--165, 2016.

\bibitem[CKN{\etalchar{+}}19]{chan2019characterization}
Timothy Chan, Daniel Kral, Jonathan~A Noel, Yanitsa Pehova, Maryam Sharifzadeh,
  and Jan Volec.
\newblock Characterization of quasirandom permutations by a pattern sum.
\newblock {\em arXiv preprint arXiv:1909.11027}, 2019.

\bibitem[CLP11]{chan2011orthogonal}
Timothy~M Chan, Kasper~Green Larsen, and Mihai P{\u{a}}tra{\c{s}}cu.
\newblock Orthogonal range searching on the {RAM}, revisited.
\newblock In {\em Proceedings of the twenty-seventh annual symposium on
  Computational geometry}, pages 1--10. ACM, 2011.

\bibitem[CLRS09]{cormen2009introduction}
Thomas~H Cormen, Charles~E Leiserson, Ronald~L Rivest, and Clifford Stein.
\newblock {\em Introduction to algorithms}.
\newblock MIT press, 2009.

\bibitem[Coo04]{cooper2004quasirandom}
Joshua~N Cooper.
\newblock Quasirandom permutations.
\newblock {\em Journal of Combinatorial Theory, Series A}, 106(1):123--143,
  2004.

\bibitem[Coo06]{cooper2006permutation}
Joshua~N Cooper.
\newblock A permutation regularity lemma.
\newblock {\em the Electronic Journal of Combinatorics}, 13(1):22, 2006.

\bibitem[CP10]{chan2010counting}
Timothy~M Chan and Mihai P{\u{a}}tra{\c{s}}cu.
\newblock Counting inversions, offline orthogonal range counting, and related
  problems.
\newblock In {\em Proceedings of the twenty-first annual ACM-SIAM symposium on
  Discrete Algorithms}, pages 161--173. Society for Industrial and Applied
  Mathematics, 2010.

\bibitem[Cro66]{crouse1966distribution}
Casparus~F Crouse.
\newblock Distribution free tests based on the sample distribution function.
\newblock {\em Biometrika}, 53(1/2):99--108, 1966.

\bibitem[CW92]{chang1992efficient}
Maw-Shang Chang and Fu-Hsing Wang.
\newblock Efficient algorithms for the maximum weight clique and maximum weight
  independent set problems on permutation graphs.
\newblock {\em Information Processing Letters}, 43(6):293--295, 1992.

\bibitem[DCVO08]{deBerg2008}
Mark De{ Berg}, Otfried Cheong, Marc Van{ Kreveld}, and Mark Overmars.
\newblock Orthogonal range searching.
\newblock {\em Computational Geometry: Algorithms and Applications}, pages
  95--120, 2008.

\bibitem[DDB16]{dhar2016study}
Subhra~Sankar Dhar, Angelos Dassios, and Wicher Bergsma.
\newblock A study of the power and robustness of a new test for independence
  against contiguous alternatives.
\newblock {\em Electronic Journal of Statistics}, 10(1):330--351, 2016.

\bibitem[DHS18]{drton2018high}
Mathias Drton, Fang Han, and Hongjian Shi.
\newblock High dimensional independence testing with maxima of rank
  correlations.
\newblock {\em arXiv preprint arXiv:1812.06189}, 2018.

\bibitem[DW80]{de1980cramer}
Tertius De~Wet.
\newblock Cram{\'e}r-von {M}ises tests for independence.
\newblock {\em Journal of Multivariate Analysis}, 10(1):38--50, 1980.

\bibitem[EL19]{even2019corners}
Chaim Even{-Zohar} and Calvin Leng.
\newblock {\em GitHub repository: \emph{corners}}, 2019.
\newblock \url{http://github.com/chaim-e/corners}.

\bibitem[Eve18]{even2018patterns}
Chaim Even{-Zohar}.
\newblock Patterns in random permutations.
\newblock {\em arXiv preprint arXiv:1811.07883}, 2018.

\bibitem[FL82]{fisher1982nonparametric}
Nicholas~I Fisher and Alan~J Lee.
\newblock Nonparametric measures of angular-angular association.
\newblock {\em Biometrika}, pages 315--321, 1982.

\bibitem[Ful04]{fulman2004stein}
Jason Fulman.
\newblock Stein's method and non-reversible {M}arkov chains.
\newblock In {\em Stein's Method}, pages 66--74. Institute of Mathematical
  Statistics, 2004.

\bibitem[FW18]{fox2018fast}
Jacob Fox and Fan Wei.
\newblock Fast property testing and metrics for permutations.
\newblock {\em Combinatorics, Probability and Computing}, 27(4):539--579, 2018.

\bibitem[GGKK15]{glebov2015finitely}
Roman Glebov, Andrzej Grzesik, Tereza Klimo{\v{s}}ov{\'a}, and Daniel
  Kr{\'a}l'.
\newblock Finitely forcible graphons and permutons.
\newblock {\em Journal of Combinatorial Theory, Series B}, 110:112--135, 2015.

\bibitem[GM14]{guillemot2014finding}
Sylvain Guillemot and D{\'a}niel Marx.
\newblock Finding small patterns in permutations in linear time.
\newblock In {\em Proceedings of the twenty-fifth annual ACM-SIAM symposium on
  Discrete algorithms}, pages 82--101. SIAM, 2014.

\bibitem[GV09]{guillemot2009pattern}
Sylvain Guillemot and St{\'e}phane Vialette.
\newblock Pattern matching for 321-avoiding permutations.
\newblock In {\em International Symposium on Algorithms and Computation}, pages
  1064--1073. Springer, 2009.

\bibitem[H{\"a}s02]{hasto2002packing}
Peter~A H{\"a}st{\"o}.
\newblock The packing density of other layered permutations.
\newblock {\em Electronic Journal of Combinatorics}, 9(2):1, 2002.

\bibitem[HH16]{heller2016computing}
Yair Heller and Ruth Heller.
\newblock Computing the {B}ergsma {D}assios sign-covariance.
\newblock {\em arXiv preprint arXiv:1605.08732}, 2016.

\bibitem[HHG12]{heller2012consistent}
Ruth Heller, Yair Heller, and Malka Gorfine.
\newblock A consistent multivariate test of association based on ranks of
  distances.
\newblock {\em Biometrika}, 100(2):503--510, 2012.

\bibitem[HHK{\etalchar{+}}16]{heller2016consistent}
Ruth Heller, Yair Heller, Shachar Kaufman, Barak Brill, and Malka Gorfine.
\newblock Consistent distribution-free k-sample and independence tests for
  univariate random variables.
\newblock {\em The Journal of Machine Learning Research}, 17(1):978--1031,
  2016.

\bibitem[HHMW20]{hartung2020combinatorial}
Elizabeth Hartung, Hung~P Hoang, Torsten M{\"u}tze, and Aaron Williams.
\newblock Combinatorial generation via permutation languages.
\newblock In {\em Proceedings of the Fourteenth Annual ACM-SIAM Symposium on
  Discrete Algorithms}, pages 1214--1225. SIAM, 2020.

\bibitem[HJD06]{harrell2006hmisc}
Frank~E Harrell~Jr and Maintainer~Charles Dupont.
\newblock The {H}misc package.
\newblock {\em R package version}, 3(0-12):3, 2006.

\bibitem[HKM{\etalchar{+}}13]{hoppen2013limits}
Carlos Hoppen, Yoshiharu Kohayakawa, Carlos~Gustavo Moreira, Bal{\'a}zs
  R{\'a}th, and Rudini~Menezes Sampaio.
\newblock Limits of permutation sequences.
\newblock {\em Journal of Combinatorial Theory, Series B}, 103(1):93--113,
  2013.

\bibitem[HKMS11]{hoppen2011testing}
Carlos Hoppen, Yoshiharu Kohayakawa, Carlos~Gustavo Moreira, and Rudini~Menezes
  Sampaio.
\newblock Testing permutation properties through subpermutations.
\newblock {\em Theoretical Computer Science}, 412(29):3555--3567, 2011.

\bibitem[Hoe48]{hoeffding1948non}
Wassily Hoeffding.
\newblock A non-parametric test of independence.
\newblock {\em The annals of mathematical statistics}, pages 546--557, 1948.

\bibitem[Hof17]{hofer2017central}
Lisa Hofer.
\newblock A central limit theorem for vincular permutation patterns.
\newblock {\em arXiv preprint 1704.00650}, 2017.

\bibitem[HS14]{han2014parallel}
Yijie Han and Sanjeev Saxena.
\newblock Parallel algorithms for testing length four permutations.
\newblock In {\em 2014 Sixth International Symposium on Parallel Architectures,
  Algorithms and Programming}, pages 81--86. IEEE, 2014.

\bibitem[HS18]{han2018algorithms}
Yijie Han and Sanjeev Saxena.
\newblock Algorithms for testing occurrences of length 4 patterns in
  permutations.
\newblock {\em Journal of Combinatorial Optimization}, 35(1):189--208, 2018.

\bibitem[HWC13]{hollander2013nonparametric}
Myles Hollander, Douglas~A Wolfe, and Eric Chicken.
\newblock {\em Nonparametric statistical methods}, volume 751.
\newblock John Wiley \& Sons, 2013.

\bibitem[Iba97]{ibarra1997finding}
Louis Ibarra.
\newblock Finding pattern matchings for permutations.
\newblock {\em Information Processing Letters}, 61(6):293--295, 1997.

\bibitem[Jan84]{janson1984asymptotic}
Svante Janson.
\newblock The asymptotic distributions of incomplete {U}-statistics.
\newblock {\em Probability Theory and Related Fields}, 66(4):495--505, 1984.

\bibitem[JK17]{jelinek2017hardness}
V{\'\i}t Jel{\'\i}nek and Jan Kyn{\v{c}}l.
\newblock Hardness of permutation pattern matching.
\newblock In {\em Proceedings of the Twenty-Eighth Annual ACM-SIAM Symposium on
  Discrete Algorithms}, pages 387--396. Society for Industrial and Applied
  Mathematics, 2017.

\bibitem[JNZ15]{janson2015asymptotic}
Svante Janson, Brian Nakamura, and Doron Zeilberger.
\newblock On the asymptotic statistics of the number of occurrences of multiple
  permutation patterns.
\newblock {\em Journal of Combinatorics}, 6:117--143, 2015.

\bibitem[Ken38]{kendall1938new}
Maurice~G Kendall.
\newblock A new measure of rank correlation.
\newblock {\em Biometrika}, 30(1):81--93, 1938.

\bibitem[Kit05]{kitaev2005partially}
Sergey Kitaev.
\newblock Partially ordered generalized patterns.
\newblock {\em Discrete Mathematics}, 298(1-3):212--229, 2005.

\bibitem[Kit11]{kitaev2011patterns}
Sergey Kitaev.
\newblock {\em Patterns in Permutations and Words}.
\newblock Springer Science \& Business Media, 2011.

\bibitem[KK14]{klimovsova2014hereditary}
Tereza Klimo{\v{s}}ov{\'a} and Daniel Kr{\'a}l'.
\newblock Hereditary properties of permutations are strongly testable.
\newblock In {\em Proceedings of the Twenty-Fifth Annual ACM-SIAM Symposium on
  Discrete Algorithms}, pages 1164--1173. Society for Industrial and Applied
  Mathematics, 2014.

\bibitem[KKRW19]{kenyon2019permutations}
Richard Kenyon, Daniel Kr\'{a}l', Charles Radin, and Peter Winkler.
\newblock Permutations with fixed pattern densities.
\newblock {\em Random Structures \& Algorithms}, doi:10.1002/rsa.20882:1--31,
  2019.

\bibitem[Kni66]{knight1966computer}
William~R Knight.
\newblock A computer method for calculating {K}endall's tau with ungrouped
  data.
\newblock {\em Journal of the American Statistical Association},
  61(314):436--439, 1966.

\bibitem[Knu68]{knuth1968fundamental}
Donald~E Knuth.
\newblock Fundamental algorithms, volume 1 of the art of computer programming.
\newblock {\em Addision Wesley, Reading, MA}, 1968.

\bibitem[Koz19]{kozma2019faster}
L{\'a}szl{\'o} Kozma.
\newblock Faster and simpler algorithms for finding large patterns in
  permutations.
\newblock {\em arXiv preprint arXiv:1902.08809}, 2019.

\bibitem[KP13]{kral2013quasirandom}
Daniel Kr{\'a}l' and Oleg Pikhurko.
\newblock Quasirandom permutations are characterized by 4-point densities.
\newblock {\em Geometric and Functional Analysis}, 23(2):570--579, 2013.

\bibitem[Mat94]{matouvsek1994geometric}
Ji{\v{r}}{\'\i} Matou{\v{s}}ek.
\newblock Geometric range searching.
\newblock {\em ACM Computing Surveys (CSUR)}, 26(4):422--461, 1994.

\bibitem[MT04]{marcus2004excluded}
Adam Marcus and G{\'a}bor Tardos.
\newblock Excluded permutation matrices and the {S}tanley--{W}ilf conjecture.
\newblock {\em Journal of Combinatorial Theory, Series A}, 107(1):153--160,
  2004.

\bibitem[NRRS19]{newman2019testing}
Ilan Newman, Yuri Rabinovich, Deepak Rajendraprasad, and Christian Sohler.
\newblock Testing for forbidden order patterns in an array.
\newblock {\em Random Structures \& Algorithms}, 55(2):402--426, 2019.

\bibitem[NWD16]{nandy2016large}
Preetam Nandy, Luca Weihs, and Mathias Drton.
\newblock Large-sample theory for the {B}ergsma-{D}assios sign covariance.
\newblock {\em Electronic Journal of Statistics}, 10(2):2287--2311, 2016.

\bibitem[{OEI}19]{oeis}
{OEIS Foundation Inc.}
\newblock The on-line encyclopedia of integer sequences,
  \texttt{http://oeis.org}, 2019.

\bibitem[Pri97]{price1997packing}
Alkes~Long Price.
\newblock Packing densities of layered *patterns.
\newblock {\em Dissertation available from ProQuest. AAI9727276}, 1997.

\bibitem[PS10]{presutti2010packing}
Cathleen~Battiste Presutti and Walter Stromquist.
\newblock Packing rates of measures and a conjecture for the packing density of
  2413.
\newblock {\em Permutation Patterns}, 376:287--316, 2010.

\bibitem[Ros75]{rosenblatt1975quadratic}
Murray Rosenblatt.
\newblock A quadratic measure of deviation of two-dimensional density estimates
  and a test of independence.
\newblock {\em The Annals of Statistics}, pages 1--14, 1975.

\bibitem[Sch61]{schensted1961longest}
Craige Schensted.
\newblock Longest increasing and decreasing subsequences.
\newblock {\em Canadian Journal of Mathematics}, 13:179--191, 1961.

\bibitem[Spe04]{spearman1904proof}
Charles Spearman.
\newblock The proof and measurement of association between two things.
\newblock {\em The American Journal of Psychology}, 15(1):72--101, 1904.

\bibitem[SS85]{simion1985restricted}
Rodica Simion and Frank~W Schmidt.
\newblock Restricted permutations.
\newblock {\em European Journal of Combinatorics}, 6(4):383--406, 1985.

\bibitem[SS17]{sliacan2017improving}
Jakub Sliacan and Walter Stromquist.
\newblock Improving bounds on packing densities of 4-point permutations.
\newblock {\em arXiv:1704.02959}, 2017.

\bibitem[Sta11]{stanley2011enumerative}
Richard~P Stanley.
\newblock {\em Enumerative Combinatorics}, volume~1.
\newblock Cambridge University Press, 2011.

\bibitem[Ste10]{steingrimsson2010generalized}
Einar Steingr{\i}msson.
\newblock Generalized permutation patterns -- a short survey.
\newblock {\em Permutation patterns}, 376:137--152, 2010.

\bibitem[SV82]{shiloach1982n2log}
Yossi Shiloach and Uzi Vishkin.
\newblock An {$O(n^2\log n)$} parallel max-flow algorithm.
\newblock {\em Journal of Algorithms}, 3(2):128--146, 1982.

\bibitem[WDL16]{weihs2016efficient}
Luca Weihs, Mathias Drton, and Dennis Leung.
\newblock Efficient computation of the {B}ergsma--{D}assios sign covariance.
\newblock {\em Computational Statistics}, 31(1):315--328, 2016.

\bibitem[WDM18]{weihs2018symmetric}
Luca Weihs, Mathias Drton, and Nicolai Meinshausen.
\newblock Symmetric rank covariances: a generalized framework for nonparametric
  measures of dependence.
\newblock {\em Biometrika}, 105(3):547--562, 2018.

\bibitem[Yan70]{yanagimoto1970measures}
Takemi Yanagimoto.
\newblock On measures of association and a related problem.
\newblock {\em Annals of the Institute of Statistical Mathematics},
  22(1):57--63, 1970.

\bibitem[YS05]{yugandhar2005parallel}
V~Yugandhar and Sanjeev Saxena.
\newblock Parallel algorithms for separable permutations.
\newblock {\em Discrete Applied Mathematics}, 146(3):343--364, 2005.

\end{thebibliography}

\end{document}